\documentclass[11pt]{article}
\usepackage[totalwidth=13.0cm,totalheight=20.0cm]{geometry}
\usepackage{latexsym,amsthm,amsmath,amssymb,url}
\usepackage[ruled, linesnumbered]{algorithm2e}

\usepackage{tikz,authblk}
\usetikzlibrary{decorations.pathreplacing}

\newtheorem{corollary}{Corollary}
\newtheorem{definition}{Definition}
\newtheorem{theorem}{Theorem}

\newtheorem{proc}{Procedure}

\begin{document}

\title{Acyclicity in Edge-Colored Graphs}

\date{}

\author[1,2]{Gregory Gutin}\author[1]{Mark Jones}\author[1]{Bin Sheng}\author[1]{Magnus Wahlstr{\"o}m}
\author[3,4]{Anders Yeo}
\affil[1]{Royal Holloway, University of London, TW20 0EX, Egham, Surrey, UK}
\affil[2]{University of Haifa, Mount Carmel, Haifa, 3498838, Israel}
\affil[3]{Singapore University of Technology and Design, 8 Somapah Road 487372, Singapore} 
\affil[4]{University of Johannesburg, Auckland Park, 2006 South Africa}
\maketitle

\pagestyle{plain}

\begin{abstract}
A walk $W$ in edge-colored graphs is called properly colored (PC) if every pair of consecutive edges in $W$ is of different color.
We introduce and study five types of PC acyclicity in edge-colored graphs such that graphs of PC acyclicity of type $i$ is a proper superset of graphs of acyclicity of type $i+1$, $i=1,2,3,4.$
The first three types are equivalent to the absence of PC cycles, PC closed trails, and PC closed walks, {respectively}. 
While graphs of types 1, 2 and 3 can be recognized in polynomial time, the problem of recognizing graphs of type 4 is, somewhat surprisingly, NP-hard even for 2-edge-colored graphs (i.e., when only two colors are used). The same problem with respect to type 5 is polynomial-time solvable for all edge-colored graphs.
Using the five types, we investigate the border between intractability and tractability for the problems of finding the maximum number of internally vertex-disjoint PC paths between two vertices and the minimum number of vertices to meet all PC paths between two vertices.
\end{abstract}

\section{Introduction} \label{sec:1}
A {\em walk} in a multigraph is a sequence $W=v_1e_1v_2\dots v_{p-1}e_{p-1}v_p$ of alternating vertices and edges such that vertices $v_i$ and $v_{i+1}$ are end-vertices of edge $e_i$ for every $i\in [p-1]$.  
A walk $W$ is {\em closed} ({\em open}, respectively) if $v_1=v_p$ ( $v_1\neq v_p$, respectively). 
A {\em trail} is a walk in which all edges are distinct, a {\em path} is a non-closed walk in which all vertices are distinct, and a {\em cycle} is a closed walk where all vertices apart from the first and last ones are distinct. 

In this paper, we study properly colored walks in graphs with colored edges, which are called {\em edge-colored graphs} or {\em $c$-edge-colored graphs} when colors are taken from the set $[c]=\{1,2,\dots ,c\}$.
For 2-edge-colored graphs we use colors blue and red instead of 1 and 2. A walk $W=v_1e_1v_2\dots v_{p-1}e_{p-1}v_p$ is {\em properly colored (PC)} if edges $e_i$ and $e_{i+1}$ are of different colors for every $i\in \{1,2,\dots, p-2\}$ and, in addition, if $W$ is closed then edges $e_{p-1}$ and $e_1$ are of different colors. PC walks are of interest  in graph theory applications, e.g., in molecular biology \cite{Dor1,Dor2,Dor3,Pev} and in VLSI for compacting programmable logical arrays \cite{HuKu}. They are also of interest in graph theory itself as generalizations of walks in undirected and directed graphs. Indeed, consider the standard transformation from a directed graph $D$ into a 2-edge-colored graph $G$ by replacing every arc $uv$ of $D$ by a path with blue edge $uw_{uv}$ and red edge $w_{uv}v$, where $w_{uv}$ is a new vertex \cite{JBJGG}.  Clearly, every directed walk in $D$ corresponds to a PC walk in $G$ and vice versa. 
On the other hand, if every edge has a distinct color (or more generally, if the coloring is proper), then clearly all trails of the underlying undirected graph~$G$ are PC trails. 
There is an extensive literature on PC walks: for a detailed survey of pre-2009 publications, see Chapter 16 of \cite{JBJGG}, and more recent papers include 
\cite{ADFKMMS,FuMa,GJSWY,Lo1,Lo2,Lo3}.

The following notion of a monochromatic vertex will often be used in this paper. A vertex $v$ in an edge-colored graph $G$ is called {\em $G$-monochromatic} if all edges incident to $v$ in $G$ are of the same color. 
Clearly, a PC { closed} walk has no $G$-monochromatic vertex. 

It is well-known and trivial to prove that every undirected and directed graph with no cycles, has no closed walks either. Surprisingly, this is not the case for PC cycles and PC walks. In fact, the properties of having no PC cycles, having no PC closed trails, and having no PC closed walks, are all distinct.
%
%
%
In this paper, in order to better understand the structure of acyclic edge-colored graphs, we introduce five types of PC acyclicity as follows.

\begin{definition}
Let $G$ be a $c$-edge-colored undirected graph, $c\ge 2$. 
An ordering $v_1,v_2,\ldots,v_n$ of vertices of~$G$ is {\bf of type}
\begin{description}
\item[1] if for every $i \in [n]$, all edges from $v_i$ to each connected component of $G[\{v_{i+1},v_{i+2},\ldots,v_n\}]$ have the same color;
\item[2] if for every $i \in [n]$, all edges from $v_i$ to $\{v_{i+1},v_{i+2},\ldots,v_n\}$ which are not bridges in $G[\{v_{i},v_{i+1},\ldots,v_n\}]$  have the same color.
\item[3] if for every $i \in [n]$, all edges from $v_i$ to $\{v_{i+1},v_{i+2},\ldots,v_n\}$ have the same color;
\item [4] if for every $i \in [n]$, all edges from $v_i$ to $\{v_{i+1},v_{i+2},\ldots,v_n\}$ have the same color and
all edges from $v_i$ to $\{v_1,v_2,\ldots,v_{i-1}\}$ have the same color;
\item [5] if for every $i \in [n]$, all edges from $v_i$ to $\{v_{i+1},v_{i+2},\ldots,v_n\}$ have the same color and
all edges from $v_i$ to $\{v_1,v_2,\ldots,v_{i-1}\}$ have the same color but different from the color of edges from $v_i$ to $\{v_{i+1},v_{i+2},\ldots,v_n\}$.
\end{description}
\end{definition}

\begin{definition}
Let $i\in [5]$. $G$ is {\bf PC acyclic of type $i$} if it has an ordering $v_1,v_2,\ldots,v_n$ of vertices of type $i$. 
\end{definition}

Clearly, the class of $c$-edge-colored acyclic graphs of type $i$ contains the class  of $c$-edge-colored acyclic graphs of type $i+1$, $i\in [4]$. We will see later in the paper that the containments are proper. 
We will see that graphs of the first two types coincide with edge-colored graphs without PC cycles and without PC closed trails, respectively. These two classes of edge-colored graphs were characterized by Yeo \cite{Yeo} and Abouelaoualim {\em et al.} \cite{ADFMMS}, respectively. We will prove that graphs  of PC acyclicity of type 3 are edge-colored graphs without PC walks. We are unaware of a ``nice'' characterization of edge-colored graphs of  type 4. In fact, we show that it is NP-hard to recognize graphs of this type, which is somewhat surprising as we prove that recognition of all other types is polynomial-time solvable. We will prove that an edge-colored graph is acyclic of type 5 if and only if every vertex is incident to edges of at most two colors and every cycle $C$ has a positive even number of vertices incident, in $C$, to edges of the same color. 
For $2$-edge-colored graphs this is equivalent to being bipartite with no PC cycle. Therefore for $2$-edge-colored graphs, being PC acyclic of type 5 is the same as being bipartite and PC acyclic of type 1.

Using the five types, we will investigate the border between intractability and tractability for the problems of finding the maximum number of internally vertex-disjoint PC paths between two  vertices and of finding the minimum number of vertices to eliminate all PC paths between two vertices. We will prove that both problems are NP-hard for 2-edge-colored graphs of PC acyclicity of type 3 (and thus of types 1 and 2), but polynomial time solvable for edge-colored graphs of PC acyclicity of type 4 (and 5). We will also show that while Menger's theorem does not hold in general, { even} on 2-edge-colored graphs of PC acyclicity of type 3 (or 1 or 2), it holds on edge-colored graphs of PC acyclicity of type 4 (and 5).

The rest of the paper is organized as follows.  In Section \ref{sec:Types}, we study the five types of acyclicity of edge-colored graphs. Section \ref{sec:Paths} is devoted to PC paths and separators in edge-colored graphs. Finally, in Section \ref{sec:Prob} we discuss an open problem.


\section{Types of PC Acyclic Edge-Colored Graphs}\label{sec:Types}

In this section,  we study the five types of PC acyclicity introduced in the previous section.

The fact that a $c$-edge-colored graph $G$ is PC acyclic of type 1 if and only if $G$ has no PC cycle follows immediately from a theorem by Yeo \cite{Yeo} below (a special case of Yeo's theorem for $c=2$ was obtained by Grossman and H{\"a}ggkvist \cite{GrHa}).

\begin{theorem}\label{Yeo} 
If a $c$-edge-colored graph $G$ has no PC cycle then $G$ has a vertex $z$ such that every connected component of $G-z$ is joined to $z$ by edges of the same color.
\end{theorem}

We will prove the following easy consequence of Theorem \ref{Yeo}.

\begin{corollary}\label{cor1}
A $c$-edge-colored graph $G$ is PC acyclic of type 1 if and only if $G$ has no PC cycle. 
\end{corollary}
\begin{proof}
Suppose that $G$ is PC acyclic of type 1 and has a PC cycle $C$. 
Consider an acyclic ordering of $V(G)$ of type 1. Let $x$ be the vertex on $C$ with lowest subscript in the acyclic ordering.
Observe that $C-x$ must belong to the same component in $G-x$ and all vertices of $C-x$ come after $x$ in the acyclic ordering. Thus, $x$ must have both incident  edges in the cycle of the same color, a contradiction.  

Now let $G$ have no PC cycle. By Theorem \ref{Yeo}, $G$ has a vertex $z$ such that every connected component of $G-z$ is joined to $z$ by edges of the same color. Set $v_1=z$ and consider $G-z$ to obtain $v_2,\dots ,v_n$. Clearly, the resulting ordering is PC acyclic of type 1.
\end{proof}
Theorem \ref{Yeo} implies that $G$ has no PC cycle if and only if $G$ has a vertex $z$ such that every connected component of $G-z$ is joined to $z$ by edges of the same color and $G-z$ has no PC cycle. Thus, checking PC acyclicity of type 1 can be done in polynomial time.

Using the next theorem, similarly to proving  Corollary \ref{cor1} we can show that a $c$-edge-colored graph $G$ is PC acyclic of type 2 if and only if $G$ has no PC closed trail.

\begin{theorem}[{Abouelaoualim {\em et al.} \cite{ADFMMS}}]
Let $G$ be a $c$-edge-colored graph, such that every vertex of $G$ is incident with at least two edges of different colors.
Then either $G$ has a bridge or $G$ has a PC closed trail.
\end{theorem}
This theorem implies that to check whether $G$ has a PC closed trail, it suffices to recursively delete all bridges and all $G$-monochromatic vertices. Observe that $G$ has a PC trail if and only if the resulting graph is non-empty. This implies that checking PC acyclicity of type 2 can be done in polynomial time.

To see that containment is proper between  acyclicities of type 1 and type 2, 
{consider a graph with vertex set $\{v_1, v_2, x, u_1, u_2\}$ and edge set $\{v_1v_2, v_1x, v_2x$, $u_1u_2, u_1x, u_2x\}$, where $v_1v_2, u_1x, u_2x$ are colored red, $u_1u_2, v_1x, v_2x$ are colored blue. Clearly, this 2-edge-colored graph $G$ has no PC cycle, but it has a PC closed trail. Thus, $G$ is acyclic of type 1 but not acyclic of type 2. }

It seems $c$-edge-colored graphs $G$ without PC closed walks have not been studied in the literature. Here is a counterpart of Theorem \ref{Yeo} for such graphs. 

\begin{theorem}\label{th3}
If a $c$-edge-colored graph $G$ has no PC closed walk then $G$ has a $G$-monochromatic vertex.
\end{theorem}
\begin{proof}
We call a $c$-edge-colored graph $H$ an {\em extension} of a $c$-edge-colored graph $G$ if $H$ is obtained from $G$ by replacing every vertex $u$ by a set $I_u$ of independent vertices with the same adjacencies and {edge} colors as $u.$
Observe that $G$ has no PC closed walk if and only if no extension of $G$, in which $I_u$ is sufficiently large, has a PC cycle. Now apply Theorem \ref{Yeo} to an extension $H$ of a connected $c$-edge-colored graph $G$ in which $|I_u|>1$ for each $u\in V(G)$, and note that for every vertex $z\in V(H)$, $H-z$ is connected.
\end{proof}

Using Theorem \ref{th3}, similarly to proving  Corollary \ref{cor1} we can show that a $c$-edge-colored graph $G$ is PC acyclic of type 3 if and only if $G$ has no PC closed walk.
Theorem \ref{th3} implies that $G$ has no PC closed walk if and only if $G$ has a vertex $z$ incident with edges of the same color and $G-z$ has no PC closed walk. Thus, checking PC acyclicity of type 3 can be done in polynomial time.

To see that containment is proper between  acyclicities of type 2 and type 3, consider the following graph $G$ with $V(G)=\{a_1,a_2,a_3,b_1,b_2,b_3\}$, with blue edges $a_1 b_1$, $a_2 b_2$ and $a_3 b_3$ and red edges
$a_1 a_2$, $b_1 a_2$, $a_3 b_2$ and $b_3 b_2$. In $G$ we have a PC closed walk $a_1 a_2 b_2 b_3 a_3 b_2 a_2 b_1 a_1$.  This walk uses the edge $a_2 b_2$ twice.  There is no PC closed trail in $G$: as $a_2 b_2$ is a bridge it does not belong to a closed trail and removing $a_2 b_2$ makes it obvious that there is no PC closed trail in the remainder.

There is unlikely to be a 'nice' characterization of $c$-edge-colored graphs of type 4 due to the following somewhat surprising result. 

\begin{theorem}
It is NP-complete to decide whether a 2-edge-colored graph is acyclic of type 4.
\end{theorem}
\begin{proof}
It is easy to see that our problem is in NP.

To prove NP-hardness, we reduce from the {\sc Betweenness} problem. In this problem, we are given a set of distinct ordered triples of elements from a universe $U.$ Our task is to decide whether there exists a linear ordering of $U,$ such that for every triple $(x,y,z)$ of distinct elements of $U$ in the input, either $x>y>z$ or $z>y>x$ in the ordering (i.e., $y$ must appear between $x$ and $z$). Then we say that each triple is {\em satisfied} by the ordering of $U$. This problem is NP-complete \cite{Opa}. 

In the rest of the proof, it will be convenient for us to write an ordering  $v_1,v_2,\dots, v_p$ of some vertices in a graph as $v_1>v_2>\dots > v_p$.

Given an instance of {\sc Betweenness}, we produce a 2-edge-colored graph $G$ as follows.
We add each element in $U$ as a vertex of $G.$ (These vertices will all be incident  only to blue edges in the final graph.)
For each triple $(x,y,z)$, we create a gadget with vertices $x,y,z$ and new vertices $a(x,y), b(x,y), b(z,y)$, $a(z,y).$
Add blue edges $xa(x,y)$, $b(x,y)b(z,y)$, $za(z,y)$, $yb(x,y)$ and $yb(z,y)$, and red edges $a(x,y)b(x,y)$, $b(z,y)a(z,y)$, see Figure \ref{fig:Example}.

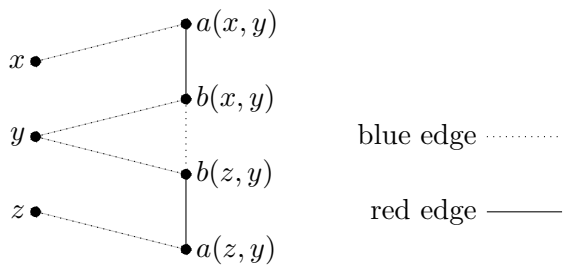
\begin{figure}\centering
\begin{tikzpicture}
\draw [dotted](5,5)[fill]circle [radius=0.07]node [left]{$z$}--(7,4.5)[fill]circle [radius=0.07] node [right]{$a(z,y)$};
\draw [dotted](5,6)circle [radius=0.07] node [left]{$y$}--(7,5.5)[fill]circle [radius=0.07]node [right]{$b(z,y)$};
\draw [dotted](5,6)--(7,6.5)[fill]circle [radius=0.07]node [right]{$b(x,y)$};
\draw [dotted](5,7)[fill]circle [radius=0.07]node [left]{$x$}--(7,7.5)[fill]circle [radius=0.07]node [right]{$a(x,y)$};
\draw [dotted](7,6.5)--(7,5.5);
\draw [](7,5.5)--(7,4.5);
\draw [](7,7.5)--(7,6.5);

 \draw [](11,5)node[left]{red edge}  --(12,5) ;
 \draw [dotted](11,6)node[left]{blue edge}  --(12,6) ;

\end{tikzpicture}
\caption{(x,y,z)-gadget}
 \label{fig:Example}
\end{figure}

This concludes the construction of $G.$

Now suppose $G$ is acyclic of type 4 and consider the ordering within the gadget for a triple $(x,y,z)$.
As $xa(x,y)b(x,y)b(z,y)a(z,y)z$ is a PC path, any acyclic ordering of type 4 must have either $x>a(x,y)>b(x,y)>b(z,y)>a(z,y)>z$ or $z>a(z,y)>b(z,y)>b(x,y)>a(x,y)>x.$
Suppose the former.
Then $y$ cannot appear before $b(x,y)$ because of the red edge $a(x,y)b(x,y),$ and it cannot appear after $b(z,y)$ because of the red edge $b(z,y)a(z,y).$ Thus $y$ must appear between $b(x,y)$ and $b(z,y),$ and in particular $y$ must appear between $x$ and $z$. A similar argument holds when $z>a(z,y)>b(z,y)>b(x,y)>a(x,y)>x.$

Thus, if $G$ is acyclic of type 4, then there is an ordering such that for every input triple $(x,y,z)$, $y$ appears between $x$ and $z$, and so our Betweenness instance is a Yes-instance.

Conversely, suppose our Betweenness instance is a Yes-instance, consider an  ordering of $U$ satisfying every triple of the instance. We extend this to an ordering of $V(G)$. For each triple $(x,y,z)$, if $x>y>z,$ then we set $x>a(x,y)>b(x,y)>y>b(z,y)>a(z,y)>z.$ If $z>y>x,$ then we set $z>a(z,y)>b(z,y)>y>b(x,y)>a(x,y)>x.$
As $a(x,y),b(x,y),b(z,y),a(z,y)$ are involved only in the $(x,y,z)$-gadget, there is an ordering of $V(G)$ that satisfies all the above orderings. It is easy to check that each of $a(x,y),b(x,y),b(z,y),a(z,y)$ is satisfied (in the sense of having all edges to earlier vertices the same color, and all edges to later vertices the same color). As all vertices from $U$ are only incident  to blue edges, they are satisfied by this ordering as well. Thus we have that $G$ is acyclic of type 4, as required.
\end{proof}

To see that containment is proper between  acyclicities of type 3 and type 4, consider a complete graph on three vertices with two blue edges and one red edge. It is easy to find a PC acyclic ordering of type 3 and to see that there is no PC acyclic ordering of type 4.

One can check whether a connected edge-colored graph $G$ is acyclic of type 5 
{ using the following:}

\begin{proc}\label{proc1}
First check that for each vertex in the graph, its incident edges are colored with at most two colors, as otherwise, $G$ cannot be acyclic of type 5. { 
Choose an arbitrary vertex $x$ and orient edges of one color out of $x$ and edges of the other color towards $x$.}
Then for every vertex $y$ for which there is an arc towards (out of, respectively) $y$ which was an edge of color $i$, mark $y$ and
orient all edges of color $i$ incident to $y$ towards (out of, respectively) $y$ and all other edges incident  to $y$ out of $y$ (towards $y$, respectively). 
{ Stop the procedure if orienting edges incident to $y$ leads to one of the following conflicts for another vertex $z$:

\begin{description}
 \item[(a)] $z$ will have two arcs of different color oriented into it or two arcs of different color oriented out of it;
 \item[(b)] $z$ will have an arc into it and an arc out of it of the same color.
\end{description}
}
\end{proc}

Theorem \ref{char5} { will show} that if this procedure is completed (without { conflicts (a) and (b)}) and the obtained digraph is acyclic then $G$ is acyclic of type 5, and otherwise it is not. 
Note that the procedure always finishes in polynomial time.


Here is a characterization of acyclic edge-colored graphs of type 5. 
Recall that a vertex $v$ in an edge-colored graph $G$ is $G$-monochromatic if all edges incident to $v$ in $G$ are of the same color. 

\begin{theorem} \label{char5}  Let $G$ be an edge-colored graph. The following are equivalent.
  \begin{enumerate}
  \item $G$ is PC acyclic of type 5.
  \item  { Procedure \ref{proc1}} completes and the resulting digraph is acyclic.
  \item Every vertex is incident to edges of at most two colors, and every cycle $C$ in $G$ has a positive even number of $C$-monochromatic vertices.
  \end{enumerate}
\end{theorem}
\begin{proof}
We will show that the implications $3 \Rightarrow 2 \Rightarrow 1 \Rightarrow 3$ hold. 

First we will show that $2 \Rightarrow 1$.  Assume that { Procedure \ref{proc1}}  completes and the resulting digraph, $D$, is acyclic.
Let $v_1,v_2,v_3,\ldots,v_n$ be an acyclic ordering of $D$ (that is, if $v_i v_j$ is an arc of $D$ then $i<j$). By the construction of $D$ we note that all
arcs into a vertex $v_i$ have the same color and all arcs out of $v_i$ also have the same color, which is different from the arcs entering $v_i$. As this holds for
all $v_i$ we note that the ordering $v_1,v_2,v_3,\ldots,v_n$ is a PC acyclic ordering of type 5 in $G$, which proves $2 \Rightarrow 1$.

We will now show that $1 \Rightarrow 3$, so assume that ${\cal O} = v_1,v_2,v_3,\ldots,v_n$ is a PC acyclic ordering of type 5 in $G$.
By the definition of PC acyclic ordering of type 5, every vertex is incident to edges of at most two colors. Let $C$ be any cycle in $G$. 
Let $A_1$ contain the vertices, $v_i$, of $C$ where both neighbors of $v_i$ on $C$ lie after $v_i$ in the ordering ${\cal O}$. 
Let $A_2$ contain the vertices, $v_i$, of $C$ where both neighbors of $v_i$ on $C$ lie before $v_i$ in the ordering ${\cal O}$. 
Let $B = V(C) \setminus (A_1 \cup A_2)$.  That is,
$B$ contains the vertices, $v_i$, of $C$ where one neighbors of $v_i$ lies after $v_i$ in the ordering ${\cal O}$  and the other lies before.
Note that $|A_1|=|A_2|$, as the cycle changes from going 'forward' to 'backward' $|A_2|$ times and changes from going 'backward' to 'forward' $|A_1|$ times.
Furthermore $|A_1|>0$ (and $|A_2|>0$) as the vertex of $C$ with minimum index in ${\cal O}$ belongs to $A_1$.
As the $C$-monochromatic vertices on $C$ are exactly $A_1 \cup A_2$, we note that this is an even positive number ($=2|A_1|=2|A_2|$). This proves $1 \Rightarrow 3$.

We will now show that $3 \Rightarrow 2$. We will prove this by showing that if $2$ is false, then $3$ is false.
So assume that { Procedure \ref{proc1}}  either does not complete or the resulting digraph, $D$, is not acyclic.
First assume that it does complete, but the resulting digraph, $D$, is not acyclic and let $C$ be a directed cycle in $D$.
Note that in $G$ there is no $C$-monochromatic vertices, which implies that $3$ is false.

{
We may thus assume that Procedure \ref{proc1} does not complete. 
Let $v_1, v_2, v_3, \ldots$ be the order in which the vertices are considered by the procedure.
Let $v_r$ be the first vertex (that is $r$ is smallest possible) such that when orienting all edges incident to $v_r$ some other vertex, $v_k$, will have one of the two conflicts in Procedure \ref{proc1}:

\begin{description}
 \item[(a)] $v_k$ will have two arcs of different color oriented into it or two arcs of different color oriented out of it;
 \item[(b)] $v_k$ will have an arc into it and an arc out of it of the same color.
\end{description}

Note that $k > r$, since otherwise we already considered all edges incident with $v_k$ and when considering $v_r$ 
it will not orient any edges in an opposite direction to what it is already oriented (by the minimality of $r$).
Note that when we consider some $v_i$, it has an arc to or from a vertex in $\{v_1,v_2,\ldots,v_{i-1}\}$ for all $i$.
Therefore $G[v_1,v_2,\ldots,v_r]$ is connected.
By (a) or (b), $v_k$ has two edges that have already been oriented after considering the vertices $v_1,v_2,\ldots,v_r$ (and form a conflict as in (a) or (b)), 
namely $v_rv_k$ and $v_jv_k$ for some $j<r$.
For a path $P$ between $v_r$ and $v_j$ in $G[v_1,v_2,\ldots,v_r]$, let $C$ be the cycle obtained by adding the edges $v_rv_k$ and $v_kv_j$ to $P$.
Let $A_1$ be all vertices on $C$, where both edges in $C$ are oriented into the vertex and 
let $A_2$ be all vertices on $C$, where both edges in $C$ are oriented out the vertex.
Note that all vertices in $V(C) \setminus \{v_k\}$ are $C$-monochromatic if and only if they belong to $A_1 \cup A_2$. Furthermore 
$v_k$ is $C$-monochromatic if and only if it does not belong to $A_1 \cup A_2$. Observe that $|A_1|=|A_2|$, so $|A_1 \cup A_2|$ is even, which implies that
there are an odd number of $C$-monochromatic vertices on $C$, and so $3$ does not hold.  Therefore $3 \Rightarrow 2$.
}
\end{proof}

Theorem \ref{char5} implies { the following simpler characterization for the case $c=2.$}

\begin{corollary} \label{prop5} 
A 2-edge-colored graph $G$ is PC acyclic of type 5 if and only if it is bipartite and has no PC cycle. 
\end{corollary}
\begin{proof} 
Observe that every cycle $C$ in a 2-edge-colored graph has an even number of $C$-non-monochromatic vertices, since the edge color must change an even number of times as we go around the cycle back to the starting point. Therefore a 2-edge-colored graph contains an odd cycle if and only if it contains a cycle $C$ with an odd number of $C$-monochromatic vertices. The result now follows by Theorem~\ref{char5}.
%
%
\end{proof}

 To see that containment is proper between acyclicities of type 4 and type 5, consider any non-bipartite 2-edge-colored graph with all edges being blue.

\section{PC Paths and Separators}\label{sec:Paths}

This section will be devoted to separators and PC paths in edge-colored graphs. 
We consider two problems. Let a $c$-edge-colored graph $G$ and distinct vertices $x, y \in V(G)$ be given.
The \emph{minimum PC separator} problem is to find a minimum-size set $S\subseteq V(G)\setminus \{x,y\}$ such that there is no PC path between $x$ and $y$ in $G-S$; the \emph{PC path packing} problem is to find the maximum number of internally vertex-disjoint PC paths between $x$ and $y$ in $G$. We will see that both problems are NP-hard on graphs $G$ that are PC acyclic of types 1, 2 or 3, even for~$c=2$. We also find that Menger's theorem fails to hold for these graphs.
On the other hand, we show that the analogue of Menger's theorem does hold for graphs~$G$ that are PC acyclic of type 4, and that both the minimum PC separator problem and the PC path packing problem are in P on these graphs, even if no acyclic ordering of type 4 is given. 
These results hold for arbitrary~$c$. 

We begin with the hardness results. 

\begin{theorem}\label{thm:noseparator}
The minimum PC separator problem is NP-hard for 2-edge-colored graphs which are acyclic of type 3. 
\end{theorem}
\begin{proof}
We give a reduction from the vertex cover problem which is to find a minimum size vertex cover of a given graph. Given an instance $H$ of the vertex cover problem, we construct a 2-edge colored graph $G$ which is PC acyclic of type 3, and two distinct vertices $x, y$, such that a set $S$ of vertices is a vertex cover of $H$ if and only if there is no PC path between $x$ and $y$ in $G-S$.

Let $V(G)=V(H)\cup\{x,y\}$ and $E(G)=E(H)\cup\{xu: {u \in V(H)}\}\cup \{vy:  {v \in V(H)}\}$, and let us color all edges in $E(H)$ red and all the edges incident to $x$ or $y$ blue. It is easy to see that graph $G$ is PC acyclic of type 3, just put $x$ and $y$ in the beginning of the vertex ordering. 

Observe that for any vertex set $S\subseteq V(H)= V(G)\setminus \{x,y\}$, there is a PC path between $x$ and $y$ in graph $G-S$ if and only if there is at least one edge in the graph $H-S$. Thus a vertex set $S$ is a vertex cover of $H$ if and only if there is no PC path between $x$ and $y$ in $G-S$. Thus we have given a polynomial reduction from the vertex cover problem to our problem, which implies the NP-completeness. 
\end{proof}

\begin{theorem}\label{thm:nopacking}
The PC path packing problem is NP-hard for 2-edge-colored graphs which are acyclic of type 3. 
\end{theorem}
\begin{proof}
We will give a reduction from the following problem called the restricted bipartite perfect matching problem: Given a bipartite graph $G$ and a partition of its edges into sets of size at most 2, decide whether $G$ has a perfect matching containing at most one edge from each partition set (we will call such a perfect matching {\em restricted}). This problem was proved to be NP-complete by Plaisted and Zaks \cite{PlZa}.

Let $(G,{\cal S})$ be an instance of the restricted bipartite perfect matching problem, where $G=(V_{1}\cup V_2, E)$ is a bipartite graph with $|V_1|=|V_2|$, and ${\cal S}$ is the collection of partite sets of size 2. Note that we may require that for every size-2 set~$\{e_i,e_j\} \in {\mathcal{S}}$ that~$e_i$ and~$e_j$ are vertex-disjoint; if not, we may split~$\{e_i,e_j\}$ into two sets~$\{e_i\}, \{e_j\}$ of size~1, as no matching can use both edges simultaneously.

We construct a PC acyclic 2-edge colored graph $H$ of type 3 in the following way. Let $E_1$ and $E_2$ denote the sets of blue and red edges of $H$, respectively. 
Introduce two new vertices $x, y$, add all edges from $\{xu: u\in V_{1}\}\cup \{vy: v\in V_2\}$ to $E_1$. For each edge $uv$ of $E$ which is not in any 2-size set of $S$, we add $uv$ to $E_{2}$. For any 2-size set $S=\{u_{i}v_{j}, u_{k}v_{l}\}$ of $\cal S$, where $\{u_{i}, u_{k}\}\subseteq V_1$, $\{v_{j}, v_{l}\}\subseteq V_2$, we add new vertices $p_S, q_S$ to $V(H)$, add edges $u_{i}p_{S}, v_{l}p_{S}, u_kq_S, v_{j}q_S$ to $E_2$, and add edge $p_{S}q_{S}$ to $E_1$; see Figure~\ref{fig:Example2}. This completes the construction of $H$. Observe that any ordering of $V(H)$ starting with $x$ and $y$ then containing all vertices of $G$ and finally having all other vertices, is PC acyclic of type 3.

\begin{figure}\centering
\begin{tikzpicture}
\draw [dotted](10,4.5)[fill]circle [radius=0.07]node [below]{$q_S$}--(7,4.5)[fill]circle [radius=0.07] node [below]{$p_S$};

\draw [dotted](7,6)circle [radius=0.07] node [left]{$x$}--(8,5.5)[fill]circle [radius=0.07]node [above]{$u_k$};
\draw [dotted](7,6)--(8,6.5)[fill]circle [radius=0.07]node [above]{$u_i$};
\draw [dotted](9,6.5)[fill]circle [radius=0.07]node [above]{$v_j$};
\draw (9,5.5)[fill]circle [radius=0.07]node [above]{$v_l$};
\draw [](8,5.5)--(10,4.5);
\draw [](9,5.5)--(7,4.5);
\draw [dotted](10,6)[fill]circle [radius=0.07] node [right]{$y$}--(9,5.5);
\draw [dotted](10,6)--(9,6.5);
\draw [](8,6.5)--(7,4.5);
\draw (9,6.5)--(10,4.5);

\draw [](3,5)[fill]circle [radius=0.07]node [left]{$u_k$}--(5,5)[fill]circle [radius=0.07] node [right]{$v_l$};

\draw [](3,6)circle [radius=0.07] node [left]{$u_i$}--(5,6)[fill]circle [radius=0.07]node [right]{$v_j$};

 \draw [](13,5)node[left]{Edge in $E_2$}  --(14,5) ;
 \draw [dotted](13,6)node[left]{Edge in $E_1$}  --(14,6) ;

\end{tikzpicture}
\caption{Construction of $H$}
 \label{fig:Example2}
\end{figure}
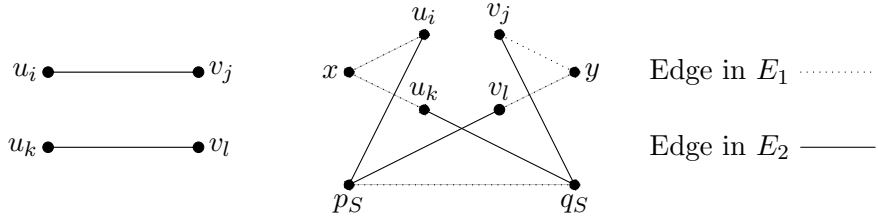

Now we prove that there is a restricted perfect matching in $G$  if and only if there are $|V_{1}|$ internally vertex-disjoint PC paths between $x$ and $y$.

Firstly, if there is a restricted perfect matching in $G$, we can easily get $|V_{1}|$ internally vertex-disjoint alternating paths between $x$ and $y$ in $H$, by constructing a PC path for each edge in the perfect matching. For each  edge $u_{i}v_{j}$ in the perfect matching, forming a 2-size set $S$ of $\cal S$ with some edge $u_{k}v_{l}$, we obtain a PC path $xu_{i}p_{S}q_{S}v_{j}y$ or $xu_iq_Sp_Sv_jy$ between $x$ and $y$. For each other edge $uv$  in the perfect matching, we have a PC path $xuvy$ between $x$ and $y$.

Secondly, if there are $|V_{1}|$ internally vertex-disjoint PC paths in $H$ between $x$ and $y$, there must be a restricted perfect matching in $G$. Observe that all the PC paths are of length 3 or 5 as all vertices in $V(G)$ are adjacent to only color 2 edges in $H-\{x,y\}.$
For each PC path we put an edge into the matching of graph $G$. Any PC path of length 3  must be of the form $xuvy$, so we put edge $uv$ into the matching. For any edge $p_{S}q_{S}$, there is at most one internally disjoint PC path passing through it, thus, if there is one such path  $xu_ip_{S}q_{S}v_jy$ ($xu_kq_{S}p_{S}v_ly$, respectively), we put edge $u_{i}v_{j}$ ($u_kv_l$, respectively) into the matching. Since we will put at most one edge from each set $S$ of $\cal S$ into the matching, the matching is a restricted matching. 
Since there are $|V_{1}|$ internally vertex-disjoint PC paths between $x$ and $y$, we will put $|V_{1}|$ non-adjacent edges into the matching, thus we obtain a perfect matching.
%
 \end{proof}

The same hardness results also hold if we consider edge-disjoint PC paths instead of vertex-disjoint ones, by a simple reduction. Namely, 
let~$G$ be a 2-edge-colored graph. We split every vertex~$v \in V(G)$ into two copies~$v'$ and~$v''$, where~$v'$ is incident with all red edges incident with~$v$, and~$v''$ with all blue edges. 
We also add a third vertex~$v_0$, a blue edge~$v'v_0$ and a red edge~$v_0v''$. 
It is easy to see that applying this transformation to every vertex in~$G$ describes a reduction from the problems described above to their edge separator, respectively edge-disjoint path packing variants.

We now show our positive result.

\begin{theorem}\label{thm3}
Let~$G$ be a~$c$-edge-colored graph which is PC acyclic of type 4, and let~$x, y$ be arbitrary distinct vertices of $G$. 
The minimum PC separator problem and the PC path packing problem for $G$ are both in P, even if no acyclic ordering for $G$ is given.
Furthermore, let $s$ be the minimum size of a subset of $V(G)\setminus \{x,y\}$ which removal eliminates all PC paths between $x$ and $y$ and 
let $t$ be the maximum number of internally vertex-disjoint PC paths between $x$ and $y$ in $G$. 
Then~$s=t$. 
\end{theorem}
\begin{proof}
  The solution to both problems, and the proof of the Menger's theorem analogue, will follow the same basic pattern. Observe that we may delete from~$G$ every monochromatic vertex distinct from~$x$ and~$y$, since such a vertex cannot be an internal vertex of any PC path. This leaves a graph, say~$H$, where every vertex except~$x$ and~$y$ is incident to edges of two colors. We show that~$H$ is PC acyclic of type 5, and that the first and last vertices of any corresponding vertex ordering are~$x$ and~$y$ (or~$y$ and~$x$), respectively. 

For this, we first note that~$x$ and~$y$ are both monochromatic in~$H$. Indeed, by assumption~$G$ is PC acyclic of type 4, which shows that there exists an induced type 4 ordering for~$H$. The first and last vertices of this ordering must be monochromatic in~$H$; hence these vertices are equal to~$x$ and~$y$. 
Furthermore, it is easy to see that for a graph with only two monochromatic vertices, the notions of being PC acyclic of type 4 and 5 coincide. Thus~$H$ is PC acyclic of type 5.

We now proceed as follows. Compute a PC acyclic ordering for~$H$ of type 5, and if necessary reverse it so that it begins with the vertex~$x$ and ends with~$y$. Note that every PC path from~$x$ to~$y$ uses the edges of~$H$ in the forward direction of the ordering only. Hence we may transform~$H$ into a digraph~$D$ by orienting every edge of~$H$ from its lower-index vertex to its higher-index vertex in the ordering, and solve the corresponding problem on~$D$ using well-known polynomial-time algorithms~\cite{JBJGG}. Menger's theorem also follows from this same reduction.
\end{proof}

Finally, we note that Menger's theorem fails to hold if~$G$ is only PC acyclic of type 3, even for~$c=2$. 
Consider the following 2-edge-colored graph $G$, see Figure \ref{fig:Example3}. Let $V(D)=\{v_{1},\ldots, v_{8}\}$, $E_1=\{v_{2}v_{3}, v_{4}v_{5}, v_{6}v_{7}\}$, $E_2=\{v_{1}v_{2}, v_{1}v_{3}, v_{1}v_{5}, v_{2}v_{6}, v_{3}v_{4}, v_{4}v_{6}, v_{5}v_{8}, v_{7}v_{8}\}$. We color edges in $E_1$ blue and edges in $E_2$ red. It is easy to check that the ordering $v_{8}v_{7}v_{6}v_{1}v_{2}v_{3}v_{4}v_{5}$ of $V(G)$ is PC acyclic of type 3. 

Let $x=v_{1}, y=v_{8}$, note that any PC path between $x$ and $y$ uses at least two blue edges, thus there is at most one internally vertex-disjoint PC path between $x$ and $y$. However, after deleting any vertex apart from $\{x, y\}$ the remaining graph will still have a PC path between $x$ and $y$: after deleting $v_{2}$ or $v_{3}$, we have PC path $xv_{5}v_{4}v_{6}v_{7}y$; after deleting $v_{4}$ or $v_{5}$, we have PC path $xv_{3}v_{2}v_{6}v_{7}y$; after deleting $v_{6}$ or $v_{7}$, we have PC path $xv_{2}v_{3}v_{4}v_{5}y$. Thus $s>t$. 

\begin{figure}\centering
\begin{tikzpicture}
\draw [](5,7.5)[fill]circle [radius=0.07]node [above]{$v_4$}--(5,4.5)[fill]circle [radius=0.07] node [below]{$v_5$};

\draw [dotted](2,6)circle [radius=0.07] node [left]{$x=v_1$}--(4,5.5)[fill]circle [radius=0.07]node [right]{$v_3$};
\draw [dotted](4,5.5)--(4,6.5)[fill]circle [radius=0.07]node [above]{$v_2$};
\draw [dotted](7,6.5)[fill]circle [radius=0.07]node [above]{$v_6$};
\draw (7,4.5)[fill]circle [radius=0.07]node [below]{$v_7$};
\draw [dotted](2,6)--(4,6.5); 
\draw [dotted](2,6)--(5,4.5);
\draw [dotted](4,5.5)--(5,7.5);
\draw [dotted](4,6.5)--(7, 6.5); 
\draw [](7,6.5)--(7,4.5);
\draw [](8,6)[fill]circle [radius=0.07] node [right]{$v_8=y$}; 
\draw [dotted](8,6)--(7,4.5); 
\draw [](4,6.5)--(4,5.5);
\draw [dotted](5,4.5)--(8,6);
\draw [dotted](5,7.5)--(7,6.5); 



\draw [dotted](12,5)node[left]{Edge in $E_2$}  --(13,5) ;
\draw [](12,6)node[left]{Edge in $E_1$}  --(13,6) ;

\end{tikzpicture}
\caption{Menger's theorem fails on $G$}
 \label{fig:Example3}
\end{figure}
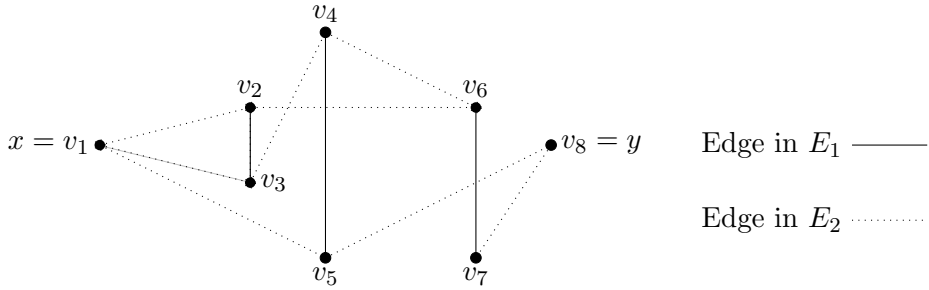

\section{Open Problem}\label{sec:Prob}

Consider the problem of deleting as few vertices as possible from a 2-edge-colored graph to get a  subgraph of PC acyclicity of type 5. We will show that this problem generalizes the directed feedback vertex set problem in digraphs and the bipartization problem. In the directed feedback vertex set problem, given a digraph $D$, find a minimum size vertex set $S$ such that $D-S$ has no directed cycle. In the bipartization problem, given an undirected graph $G$, find a minimum size vertex set $S$ such that $G-S$ is bipartite. Both directed feedback vertex problem and bipartization problem are NP-hard but fixed-parameter tractable with respect to the parameter $|S|$ \cite{CFLMPPS15,DF13}, i.e., both problems can be solved by an algorithm of running time $O(f(|S|)n^d)$, where $f$ is a computable function of $|S|$ only, $n$ is the number of vertices in the input directed or undirected graph, and $d$ is a constant. Thus, the problem of vertex deletion to a PC acyclic 2-edge-colored graph of type 5, is NP-hard, but we do not know whether it is fixed-parameter tractable with respect to the minimum size of a solution. 

For a digraph $D$, let $G$ be the $2$-edge-colored graph obtained by duplicating every vertex $v \in D$, to 
$v'$ and $v''$ and adding a red edge between $v'$ and $v''$. Then for every arc $uv$ in $D$ we add the blue edge $u'' v'$ to $G$. This completes the description of $G$. 
  Clearly $G$ is bipartite.
  If $S$ is a minimum vertex set of $G$ such that $G-S$ is PC acyclic of type $5$, then by Theorem \ref{prop5} $G-S$ has no PC cycle. By Theorem \ref{prop5} and the minimality of $S$ we may assume that $S \subseteq V(D)':=\{v':\ v\in V(D)\}$ (if $v'' \in S$ we can take $v'$ and not $v''$). This implies that $D-S$ has no directed cycle (if it did we would have a PC cycle in $G-S$). 
  Conversely, if $T$ is a minimum feedback vertex set in $D$, then $D-T$ has no directed cycle and it is not difficult to see that $G-T'$ has no PC cycle. As $G-T'$ is bipartite and has no PC cycle,
by Corollary \ref{prop5}  it is of type 5. 
 Thus, our problem generalizes the directed feedback vertex set problem in digraphs.

If $G$ is a graph, then let the edge-colored graph $G'$ be equal to $G$, where all edges are blue. If $G-S$ is bipartite then $G'-S$ is bipartite and has no PC cycle (as all edges have the same color) and therefore $G'-S$ is PC acyclic of type $5$.
  Conversely, if $G'-S$ is PC acyclic of type 5, then $G'-S$ is bipartite and $G-S$ is bipartite.
Thus, our problem generalizes   the bipartization problem.

 Of course our problem is more general, as there are many 2-edge-colored graphs that do not arise from directed graphs using the standard transformation given in Section \ref{sec:1}.  
 
 \vspace{3mm}
 
 \noindent{\bf Acknowledgment.} We are very grateful to the referees for their helpful suggestions. Research of GG was partially supported by Royal Society Wolfson Research Merit Award. Research of BS was partially supported by China Scholarship Council.

\end{document}